\documentclass[reqno,11pt]{amsart}
\makeatletter
\def\section{\@startsection{section}{1}%
	\z@{.7\linespacing\@plus\linespacing}{.5\linespacing}%
	{\bfseries
		\centering
}}
\def\@secnumfont{\bfseries}
\makeatother

\usepackage{amsmath}
\usepackage{amsfonts}
\usepackage{amssymb}
\usepackage{graphicx}
\usepackage{xcolor}
\usepackage{soul}
\usepackage{ulem}
\usepackage{lineno}
\newtheorem{theorem}{Theorem}[section]

\newtheorem{corollary}[theorem]{Corollary}

\newtheorem{definition}[theorem]{Definition}

\newtheorem{lemma}[theorem]{Lemma}

\newtheorem{proposition}[theorem]{Proposition}
\newtheorem{remark}[theorem]{Remark}

\numberwithin{equation}{section}
\usepackage{lmodern}
\usepackage[colorlinks=true, allcolors=blue]{hyperref}
\colorlet{blu1}{blue!70!black}
\colorlet{blu2}{blue!50!black}
\colorlet{blu3}{blue!70!red}
\colorlet{blu4}{blue!60!green}
\colorlet{red1}{red!80}
\colorlet{red2}{red!50!black}
\colorlet{red3}{red!70!yellow}
\colorlet{red4}{red!50!yellow}
\colorlet{yel1}{yellow!50!black}
\colorlet{yel3}{yellow!20!blue}
\colorlet{gre1}{green!60!blue}
\colorlet{gre2}{green!60!black}
\colorlet{gre3}{green!40!black}
\catcode `\@=12 \textwidth16cm \textheight23cm \hoffset-1cm   \voffset-2.0cm

\begin{document}
{\tiny\bf\color{brown}Accepted by "Stochastics: An International Journal Of Probability And Stochastic Processes"}
	\vspace{1.5cm}
\begin{center}

{\bf \Large Probabilistic Analysis of the $({\bf q},2)-$Fock Space: Vacuum Distribution and Moments of the Field Operator} \vspace{1.5cm}

{\bf\large Yungang Lu}\vspace{1cm}
	
{\large Department of Mathematics, University of Bari ``Aldo Moro''}\vspace{0.3cm}
	
{\large Via E. Orabona, 4, 70125, Bari, Italy} \vspace{0.3cm}
	
e-mail:yungang.lu@uniba.it
\end{center}\vspace{1.2cm}

\begin{center} {\large ABSTRACT}\end{center}\vspace{0.6cm}

\noindent We study some probabilistic aspects of the $({\bf q},2)-$Fock space over a given Hilbert space. Specifically, our focus is on determining  the vacuum distribution of the field operator. We ascertain the precise expression for the vacuum expectation of any arbitrary product of creation and annihilation operators, with particular emphasis on the vacuum moments of the field operator. Utilizing these results, we calculate the moment-generating function of the field operator and subsequently determine the corresponding distribution.\medskip


\section{Introduction}\label{q2pr-s1}

The $q-$commutation relation (referred to as $q-$CR for simplicity) and the $q-$Fock space were introduced as generalizations and deformations of the standard canonical commutation relation (CCR), canonical anti--commutation relation (CAR) and the Bosonic--Fermionic Fock spaces commonly used in quantum theory (see \cite{Bieden1989}, \cite{Gre91}, \cite{Gre2001}, \cite{GreHil99} and references within). They lead to various mathematical and physical applications, particularly in quantum groups, quantum field theory, statistical mechanics, and quantum information theory (see, e.g. \cite{KayLadMos2006}, \cite{KliSch1997}, \cite{Woit2024}). 
	
Closed related, the $q-$Gaussian process (and $q-$Gaussian distribution) was formally introduced in \cite{Fris-Bou70}, \cite{Gre90}, and \cite{Tsallis1988}.
It was subsequently mathematically established and investigated in the context of non-commutative probability in further works such as \cite{BoKumSpe97}, \cite{Bo-Spe91}, and \cite{Bo-Spe92}. Notably, the $q-$Gaussian distribution is characterized as the vacuum distribution of the field operator, which is the sum of creation and annihilation operators, in a $q-$deformed Fock space. 
Furthermore, the corresponding stochastic calculus was developed in \cite{DonMar2003}.

The standard $q-$CR 
\begin{align}\label{q-comm}
a_q(f)a_q^+(g)-qa_q^+(g)a_q(f)=\,\langle f,g\rangle,\qquad \forall f,g\in\mathcal{H}
\end{align}
and the $q-$Fock space have been extended by different authors. Here are two representative extensions:

$\bullet$ In \cite{Blitvic}, the standard $q-$CR \eqref{q-comm} is replaced by the following $(q,t)-$commutation relation ($(q,t)-$CR for simplicity):
\begin{align}\label{qt-01}
a_{q,t}(f)a_{q,t}^+(g)-qa_{q,t}^+(g)a_{q,t}(f)=\,\langle f,g\rangle\,t^N,\qquad \forall f,g\in\mathcal{H}
\end{align}
where, $t>0$ and $\mathcal{H}$ is a given Hilbert space; $a_{q,t}(f)$, $a_{q,t}^+(f)$ are the annihilation and creation operator with the {\it test function} $f\in \mathcal{H}$, and $N$ is the number operator on the $(q,t)-$Fock space over ${\mathcal H}$ (see \cite{Blitvic} for detail).

$\bullet$ in \cite{AsaYos2020}, motivated by the $(q,t)-$extension of $q-$CR and $q-$Fock space mentioned above, and  by a particular commutation relation introduced in \cite{aclu-qqbit} as an example of the {\bf type II} interacting Fock space, the standard $q-$CR \eqref{q-comm} is {\it weighted} as the following $(q,\{\tau_n\})-$CR:
\begin{align}\label{qt-02}
a_{q,\{\tau_n\}}(f)a_{q,\{\tau_n\}}^+(g)-q\beta_N a_{q,\{\tau_n\}}^+(g)a_{q,\{\tau_n\}}(f)=\,\langle f,g\rangle\,\tau_N,\qquad \forall f,g\in\mathcal{H}
\end{align}
where, $\{\tau_n\}_{n\in {\mathbb N}^*}\subset (0,\infty)$ and $\beta_n:=\frac{\tau_{n+1}}{\tau_n}$ for any $n\in {\mathbb N}^*$. $a_{q,\{\tau_n\}}(f)$, $a_{q,\{\tau_n\}}^+(f)$ and $N$ are the annihilation and creation operator with the {\it test function} $f\in {\mathcal H}$, and the number operator on a suitable Fock space, namely the $\tau -$weighted (generalized) $q-$Fock space over ${\mathcal H}$ (see \cite{AsaYos2020} for detail).


In \cite{YGLu2023a} (as well as \cite{YGlu2008}), the notion of $({\bf q},m)-$Fock space over a given Hilbert space ${\mathcal H}$ (refer to Section \ref{prq2-s2} for its definition in case the $m=2$) is introduced. This construction aims to illuminate the Quan--algebra, which offers a natural generalization of the traditional $q-$algebra. In this framework, creation and annihilation operators obey the $({\bf q},m)-$communication relation ($({\bf q},m)-$CR for the simplicity), serving as an analogue to the standard $q-$CR:
\begin{equation}\label{(q,m)-comm01a}
A(f)A^+(g)-	q{\bf p}_m\,A^+(g)A(f)=\langle f, g\rangle \,,\quad \forall\, f,g\in \mathcal{H}
\end{equation}
and 
\begin{equation}\label{(q,m)-comm01b}
A(f){\bf p}_{k+1}={\bf p}_kA(f)\,,\qquad \forall\, k\in \mathbb{N} \text{ and } f\in\mathcal{H}
\end{equation}
where, 

$\bullet$ for any {\it test function} $f\in{\mathcal H}$, $A(f)$ (respectively, $A^+(f)$) is the $({\bf q},m)-$annihilation (respectively, creation) operator;  in other words, $A(f)$ is the annihilation operator and $A^+(f)$ is the creation operator on the $({\bf q},m)-$Fock space over ${\mathcal H}$ (detailed in Section \ref{prq2-s2} for $m=2$);

$\bullet$ 
for any $k\in \mathbb{N}$,
\begin{equation*}
{\bf p}_k:=\begin{cases}0,&\text{ if }k=0\\
\text{the projector onto }\bigoplus_{j=0} ^{k-1}\mathcal{H}_j, &\text{ if }k\ge 1	
\end{cases}
\end{equation*} 
and each $\mathcal{H}_r$ is the {\it $r-$particle subspace} of the $({\bf q},m)-$Fock space over ${\mathcal H}$.

As mentioned in \cite{Gre2001}, \cite{GreHil99}, and related references, the $q-$CR could interpreted as a convex combination of the CCR with wight $(1+q)/2$ and the CAR with wight $(1-q)/2$. Similarly, the formula \eqref{(q,m)-comm01a} could be regarded as the convex combination of the CCR with wight $(1+q{\bf p})/2$ and CAR with wight $(1-q{\bf p})/2$, where ${\bf p}$ is a projector. 

In the present paper, our focus lies on the $({\bf q},2)-$Fock space over ${\mathcal H}$.
To highlight its novelty, we examine some combinatorial distinctions between this space and other Fock spaces mentioned earlier. Let $a(f)$  (respectively, $a^+(f)$) denote any of the following annihilation operators: $a_q(f)$, $a_{q,t}(f)$, $a_{q,\{\tau_n\}}(f)$ and $A(f)$ (respectively, the corresponding creation operators); and let $\Phi$ denote the vacuum vector for all these Fock spaces. For any $n\in\mathbb{N}^*$ and $\big\{f_k,g_k:\,k\in\{1,\ldots,n\}\big\}\subset {\mathcal H}$, the vector
\begin{align}\label{comb01}
a(f_n)\ldots a(f_1)a^+(g_1)\ldots a^+(g_n)\Phi
\end{align}
must take the form:
\begin{align}\label{comb01a}
\sum_{\sigma\in{\mathcal S}_{n}} w_{n}(q,\sigma)\prod_{k\in\{1,\ldots,n\}}\langle f_k, g_{\sigma(k)}\rangle \Phi
\end{align}
Here, ${\mathcal S}_n$ denotes the $n-$th permutation group, and $w_{n}(q, \sigma)\in{\mathbb R}$ for any $\sigma\in {\mathcal S}_n$. Moreover:

$\bullet$ For $q=0$, $w_{n}(0,\sigma)\ne0$ if and only if $\sigma$ is the identity permutation. Thus, the set $\{\sigma\in{\mathcal S}_{n}:\, w_{n}(0,\sigma)\ne0\}$ contains exactly one element.

$\bullet$ For $q\ne0$, $w_{n}(q,\sigma)\ne0$ for any $\sigma\in{\mathcal S}_{n}$ whenever the Fock space is either the standard $q-$, $(q,t)-$, or $(q,\{\tau_n\})-$Fock space. In other words, whenever the creation--annihilation operators verify either the standard $q-$CR, $(q,t)-$CR, or $(q,\{\tau_n\})-$CR, the set $\{\sigma\in {\mathcal S}_{n}:\, w_{}(q,\sigma) \ne0:\}$ exactly equals ${\mathcal S}_{n}$ and contains $n!$ elements.

$\bullet$ If $q\ne0$ and if the Fock space is the $({\bf q},2)-$Fock space (equivalently, the creation--annihilation operators verify the $({\bf q},2)-$CR), then $w_{n}(q,\sigma)\ne0$ for any $\sigma\in{\mathcal S}_n$ when $n\le 2$; in the case of $n>2$, $w_{n}(q,\sigma)\ne0$ if and only if $\sigma\in{\mathcal S}_{n,1}:=\{\sigma\in{\mathcal S}_{n}: \sigma(1)=1\}$. It is evident that ${\mathcal S}_{n,1}$ contains $(n-1)!$ elements.

To further comprehend the novelty of the $({\bf q},2)-$Fock space from a combinatorial perspective, we denote, as customary, for any $n\in{\mathbb N}^*$,
\begin{align*}
&\{-1,1\}^n:=\text{the set of all $\{-1,1\}-$valued function defined on }\{1,\ldots,n\}\\
&\{-1,1\}^{2n}_+:=\Big\{\varepsilon\in\{-1,1\}^{2n}:\,
\sum_{h=1}^{2n}\varepsilon(h)=0,\, \sum_{h=p}^{2n} \varepsilon(h)\ge0,\,\forall p\in \{1, \ldots,2n\} \Big\}\end{align*}
and 
\begin{align*}
PP(2n):=&\text{the set of all pair partitions of }\{ 1,\ldots,2n\}\\
NCPP(2n):=&\text{the set of all non--crossing pair partitions of }\{ 1,\ldots,2n\}\end{align*}
Where, recall that for any $n\in{\mathbb N}^*$,  $\{(l_h,r_h)\}_{h=1}^n$ is called a pair partition of the set $\{ 1,\ldots,2n\}$, if
\[\big\{l_h,r_h:\,h\in\{1,\ldots,n\}\big\}=\{1,\ldots,2n\};\qquad l_h<r_h,\ \forall h;\qquad l_1<\ldots<l_n\]
In particular, pair partition $\{(l_h,r_h)\}_{h=1}^n$
is {\bf non-crossing} if the following equivalence holds for any $1\le h<k\le n$:
\[l_h<l_k<r_h\,\iff\, l_h<r_k<r_h\]

Introducing further, for any $n\in{\mathbb N}^*$ and $\varepsilon\in \{-1,1\}^{2n}$,
\begin{align*}
PP(2n,\varepsilon):=&\big\{\{(l_h, r_h)\}_{h=1}^n\in PP(2n):\,\varepsilon^{-1}(\{-1\})= \{l_h:\,h\in\{1, \ldots,n\}\}\big\}\\
=& \big\{\{(l_h, r_h)\}_{h=1}^n\in PP(2n):\,\varepsilon(l_h)=-1\text{ and } \varepsilon(r_h)=1\text{ for any }h\in\{1, \ldots,n\}\big\}\end{align*}
It is well-known (see, e.g., \cite{Ac-Lu96} and \cite{Ac-Lu2022a}) that

$\bullet$ the set $PP(2n,\varepsilon)$ is empty if $\varepsilon\in\{-1,1\}^{2n}\setminus\{-1,1\}_+^{2n}$;

$\bullet$ for any $\varepsilon\in\{-1,1\}_+^{2n}$, there exists unique non--crossing element, denoted as $\theta(\varepsilon):=\{(l^\varepsilon_h,r^\varepsilon_h)\}_{h=1}^n$, in the set $PP(2n,\varepsilon)$, i.e., $PP(2n,\varepsilon)\cap NCPP(2n)= \big\{\theta(\varepsilon)\big\}$; moreover, $PP(2n,\varepsilon)$ contains exactly $\prod_{h=1}^n (2h-l^\varepsilon_h)$ elements.

Utilizing these notations and conventions, let's examine, for any $n\in{\mathbb N}^*$, $\varepsilon \in\{-1,1\} _+^{2n}$, and $\{f_1,\ldots,f_{2n}\} \subset {\mathcal H}$, the following generalization of the vector \eqref{comb01}
\begin{align}\label{comb02}
a^{\varepsilon(1)}(f_1)\ldots a^{\varepsilon(2n)}(f_{2n})\Phi
\end{align}
Here, as usual, $a^\epsilon(f):=\begin{cases} a^+(f), &\text{ if }\epsilon=1\\ a(f),&\text{ if } \epsilon=-1  \end{cases}$ for any $\epsilon\in \{-1,1\}\text{ and }f\in{\mathcal H}$.

The vector \eqref{comb02} must take the form (as discussed \cite{AsaYos2020}, \cite{Blitvic}, \cite{BoKumSpe97}, \cite{Bo-Spe91}, \cite{Bo-Spe92}, \cite{YGLu2023a}, \cite{YGLu2023c}, and related references):
\begin{align}\label{comb02a}
\sum_{\{(l_h, r_h)\}_{h=1}^n\in PP(2n,\varepsilon)} w(q,\{(l_h, r_h)\}_{h=1}^n) \prod_{h\in\{1, \ldots,2n\}}\langle f_{l_h},f_{r_h} \rangle \Phi
\end{align}
with $w(q,\{(l_h, r_h)\}_{h=1}^n)\in{\mathbb R}$ for any $\{(l_h, r_h)\}_{h=1}^n\in PP(2n,\varepsilon)$. Moreover, in the case of $q=0$, $w(0,\{(l_h, r_h)\}_{h=1}^n)\ne0$ if and only if $\{(l_h, r_h)\}_{h=1}^n=\{(l_h ^\varepsilon, r_h^\varepsilon)\}_{h=1}^n$, i.e., $\{(l_h, r_h)\}_{h=1}^n$ is the unique non--crossing pair partition in $PP(2n,\varepsilon)$.

For $q\ne0$, $w(q,\{(l_h, r_h)\}_{h=1}^n)\ne0$ for any $\{(l_h, r_h)\}_{h=1}^n\in PP(2n,\varepsilon)$ whenever the Fock space is either the standard $q-$, $(q,t)-$, or $(q,\{\tau_n\})-$Fock space.

When $q\ne0$ and the Fock space is the $({\bf q},2)-$Fock space, $w(q,\{(l_h, r_h)\}_{h=1}^n)\ne0$ if and only if $\{(l_h, r_h)\}_{h=1}^n\in {\mathcal P}_n(\varepsilon)$. The construction of ${\mathcal P}_n(\varepsilon)$ was provided \cite{YGLu2023c} as follows: By defining, for any non--crossing pair partition $\{(l_h,r_h)\}_ {h=1}^n$, the {\it depth} of the pair $(l_h, r_h)$ as the cardinality of the set $\{k:l_k<l_h<r_h<r_k\}$, a generic element of ${\mathcal P}_n(\varepsilon)$ comprises:

$\bullet$ all pairs $(l_h^\varepsilon, r_h^\varepsilon)$ with the depth greater than or equal to $2$, 

$\bullet$ an arbitrary pair partition of the set \[\{l_h^\varepsilon, r_h^\varepsilon:\, h\in\{1,\ldots,n\} \text{ and the depth of the pair }(l_h^\varepsilon, r_h^\varepsilon)<2\}\]

Consequently, when determining moments of the field operator on the $({\bf q},2)-$Fock space, the set ${\mathcal P}_n:=\bigcup_ {\varepsilon\in \{-1,1\} _+^{2n}}{\mathcal P}_n(\varepsilon)$ serves a similar role to the set $PP(2n)$ (respectively, $NCPP(2n)$) in determining moments of the field operator on the $q-$Fock space, as well as $(q,t)-$ and $(q,\{\tau_n\})-$Fock space (respectively, the full Fock space).

In the previous work \cite{YGLu2023a}, we laid the groundwork for Quan--algebra, the $*-$algebra generated by $\big\{A(f),A^+(g),{\bf p}_k:f,g\in \mathcal{H}\text{ and }k\in \mathbb{N}\big\}$ subject to the $({\bf q},m)-$CR. 
We explored some of its fundamental properties, particularly focusing on Wick's theorem.

In \cite{YGLu2023c} and \cite{YGLu2023d}, we delved into the combinatorial aspects of the $({\bf q},2)-$Fock space. The main goal of those papers was to characterize the sets ${\mathcal P}_n(\varepsilon)$ and ${\mathcal P}_n$. We also revealed a strong connection between the cardinality of these sets of pair partitions and Catalan's convolution formula introduced in \cite{Catalan1887} (for further details, see \cite{FreLac2003}, \cite{GessLac2001}, \cite{Lac2000}, \cite{Tedford2011}, and references therein; particularly, refer to \cite{Regev2012} and its cited references for various proofs).

Now, our attention turns to the probabilistic aspects of $({\bf q},2)-$Fock space. The central result of this paper, Theorem \ref{q2pr30}, provides an explicit formulation of the probability measure $L_{q,f}$, which represents the vacuum distribution of the field operator $Q(f):=A(f)+ A^+ (f)$ on the $({\bf q},2)-$Fock space over a given Hilbert space ${\mathcal H}$ with the test function $f\in {\mathcal H}$.

The symmetric nature of the distribution $L_{q,f}$, as shown in Proposition \ref{q2pr12} (i.e., all odd moments of $L_{q,f}$ are zero), ensures that $L_{q,f}$ is entirely determined by the vacuum distribution of the operator $Q(f)^2$. Furthermore, due to the boundedness of the creation and annihilation operators, as confirmed in Proposition \ref{q2pr12}, this distribution can be fully characterized by its moments or, equivalently, by its moment-generating function. Section \ref{prq2-s2} is primarily dedicated  to calculating the moment-generating function of $Q(f)^2$ with respect to the vacuum state.

Following the derivation of the moment-generating function, our attention shifts to investigating the corresponding distribution $L_{q,f}$. The main result can be roughly formulated as follows: {\it The distribution $L_{q,f}$ is:

$\bullet$ the Dirac measure centred at original if $f=0$;

$\bullet$ the two-point law on the set $\{-\Vert f\Vert,\Vert f\Vert\}$ with the equi--probability ${1\over2}$ if $f\ne0$ and $q=-1$; 

$\bullet$ the semi--circle law over the interval $( -2\Vert f\Vert ,2\Vert f\Vert ) $ if $f\ne0$ and $q=0$;

$\bullet$ an absolutely continuous probability measure with the probability density function $l_{q,f}$ provided in \eqref{q2pr30a} if $f\ne0$ and $q\in\left( -1,\frac{1}{2}\right]$;

$\bullet$ a non--absolutely continuous probability measure given in \eqref{q2pr30b} if $f\ne0$ and $q\in\left(\frac{1}{2},1\right]$.  }

\noindent In fact, among the five conclusions above, the first one (for $f=0$) is trivial, while the second the third are easily proved (see Proposition \ref{q2pr20}). Our main objective in Section \ref{prq2-s3} to prove the fourth and fifth conclusions referenced earlier. 

\section{$({\bf q},2)-$Fock space and the moment-generating function of field operator\\ with respect to the vacuum state} \label{prq2-s2}

\subsection{$({\bf q},2)-$Fock space }\label{prq2-s2-1}

Now, our attention turns to the $({\bf q},2)-$Fock space, which is a specific case of the $({\bf q},m)-$Fock space introduced in \cite{YGLu2023a}. The $({\bf q},m)-$Fock space itself is a type of Interacting Fock Space (see, e.g., \cite{acluvo-kyoto} and \cite{lu95} for a general definition and discussion).

Let $\mathcal{H}$ be a Hilbert space with a dimension greater than or equal to 2 (this assumption will be maintained throughout), equipped with the scalar product $\langle \cdot,\cdot \rangle$. 
Now, let's introduce a sequence of operators  for any $q\in[-1,1]$ as follows:

$\bullet$ The operator $\lambda_1$ is defined as the identity operator on $\mathcal{H}$, symbolized by ${\bf 1}_{\mathcal{H}}$;

$\bullet$ The operator $\lambda_{2}$ is defined on $\mathcal{H}^{\otimes 2}$ by the {\bf linearity} and the following equality:
\begin{align*}
\lambda_2(f\otimes g):=f\otimes g+q g\otimes f\,,\quad\forall f,g\in\mathcal{H}
\end{align*} 

$\bullet$ For any $n\in\mathbb{N}^*$, the operator $\lambda_{n+2}$ is defined as ${\bf 1}_{\mathcal{H}} ^{\otimes n}\otimes \lambda_2$. 

As discussed in \cite{YGLu2023a}, it is straightforward to verify the positivity of $\lambda_n$'s, and thus
\[
\mathcal{H}_n:=\text{ the completion of the } \big(\mathcal{H}^{\otimes n}, \langle \cdot,\lambda_{n} \cdot \rangle_{\otimes n}\big)/Ker\langle \cdot, \lambda_{n}\cdot \rangle_{\otimes n}
\]
forms a Hilbert space, where $\langle \cdot,\cdot \rangle_{\otimes n}$ is the standard tensor scalar product. The scalar product of $\mathcal{H}_n$ is denoted by $\langle \cdot,\cdot \rangle_{n}$, with $\langle \cdot,\cdot \rangle_{1}:=\langle \cdot,\cdot \rangle$ and
\begin{equation*}
\langle F,G \rangle_{n}:=\langle F,\lambda_nG \rangle_{\otimes n} \,,\quad\forall n\ge2\text{ and }F,G\in \mathcal{H}^{\otimes n}
\end{equation*}
or equivalently,
\begin{align*}
\langle F,G\otimes f\otimes g \rangle_{n+2}:=&\langle F,G\otimes f\otimes g\rangle_{\otimes (n+2)}+ q\langle F,G\otimes g\otimes f\rangle_{\otimes (n+2)}\notag\\
&\forall n\in \mathbb{N},\, F\in \mathcal{H}^{\otimes (n+2)},\, G\in \mathcal{H}^{\otimes n}\text{ and }f,g\in \mathcal{H}
\end{align*}

\begin{definition}\label{q2pr10c}Let $\mathcal{H}$ be a Hilbert space, let $\mathcal{H}_n$ be the Hilbert space previously defined for any $n\in \mathbb{N}$, with $\mathcal{H}_{0} :=\mathbb{C}$.

$\bullet$ The Hilbert space $\Gamma_{q,2} (\mathcal{H}) :=\bigoplus_{n=0}^{\infty}\mathcal{H}_n $ is referred to as the {\bf $({\bf q},2)-$Fock space} over $\mathcal{H}$;

$\bullet$ The vector $\Phi:=1\oplus0\oplus0 \oplus\ldots$ is termed as the {\bf vacuum vector} of $\Gamma_{q,2}(\mathcal{H})$;

$\bullet$ For any $n\in\mathbb{N}^*$, $\mathcal{H}_n$ is named as the {\bf $n-$particle space}.
\end{definition}

Throughout, we will use $\langle \cdot,\cdot\rangle$ and $\Vert\cdot\Vert$ to denote the scalar product and the induced norm, respectively, both in $\Gamma_{q,2}(\mathcal{H})$ and in $\mathcal{H}_n$'s if there is no confusion.

It is easy to observe the following {\bf consistency} of $\langle \cdot,\cdot \rangle_{n}$'s: for any $0\ne f\in\mathcal{H}$ and $n\in\mathbb{N}^*$,
\begin{equation*}
\Vert f\otimes F\Vert=0\text{ whenever } F\in\mathcal{H}_{n} \text{ verifies } \Vert F\Vert=0
\end{equation*}
This consistency guarantees that for any $f\in\mathcal {H}$, the operator which maps $F\in\mathcal{H}_{n}$ to
$f\otimes F\in \mathcal{H}_{n+1}$ is a well--defined {\bf linear} operator from $\mathcal{H}_{n}$ to $\mathcal{H}_{n+1}$.

\begin{definition}\label{q2pr11} For any $f\in \mathcal{H}$, the $({\bf q},2)-${\bf creation operator} (with the test function $f$), denoted by $A^+(f)$, is defined as the {\bf linear} operator on $\Gamma_{q,2} (\mathcal{H})$ with the following properties:
\begin{equation*}
A^+(f)\Phi:=f\,,\quad A^+(f)F:=f\otimes F,\quad \forall n\in\mathbb{N}^*\text{ and }F\in\mathcal{H}_n
\end{equation*}
\end{definition}


Recall from Section \ref{q2pr-s1} that for any $\varepsilon\in\{-1,1\}^{2n}_+$, there exists a unique non--crossing pair partition $\{(l^\varepsilon_h, r^\varepsilon_h)\} _{h=1}^n$ in $PP(2n,\varepsilon)$. This uniqueness makes sure that the map $\tau$ from $\{-1,1\}_+^{2n}$ to $NCPP(2n)$ defined by $\tau(\varepsilon):= \{(l^\varepsilon_h, r^\varepsilon_h)\}_{h=1}^n$ is a bijection. 
%
%
Moreover, by further denoting
\begin{align*}
\{-1,1\}^{2n}_{+,*}&:=\big\{\varepsilon\in\{-1,1\}_+ ^{2n},\ \sum_{h=p}^{2n}\varepsilon(h)=0\text{ only for }p=1 \big\}\notag\\
NCPP_*(2n)&:=\big\{\{(l_h,r_h)\}_{h=1}^n\in NCPP(2n):\ r_1=2n\big\}
\end{align*}
the map $\tau$ clearly induces a bijection between $\{-1,1\}^{2n}_{+,*}$ and $NCPP_*(2n)$.

\begin{remark}\label{q2pr16x} Let $n\in{\mathbb N}^*$, the following facts are easily verified.

$\bullet$ For any $\varepsilon\in \{-1,1\}^{2n}_{+,*}$, if we define $\varepsilon'(k):=\varepsilon(k+1)$ for all $k\in\{1,\ldots,2n-2\}$, then it follows that $\varepsilon'\in \{-1,1\}^{2(n-1)}_{+}$. Moreover, $\varepsilon'$ covers $\{-1,1\}^{2(n-1)}_{+}$ as $\varepsilon$ varies over $\{-1,1\}^{2n}_{+,*}$.

$\bullet$ For any $\{(l_h,r_h)\}_{h=1}^n\in NCPP_*(2n) $, if we defined $l'_h:=l_{h+1}-1$ and $r'_h:=r_{h+1}-1$ for all $h\in\{1,\ldots,n-1\}$, then it follows that $\{(l'_h,r'_h)\}_{h=1}^{n-1}\in NCPP(2(n-1))$. Moreover, $\{(l'_h,r'_h)\}_{h=1}^{n-1}$ travels over $NCPP(2(n-1))$ as $\{(l_h,r_h)\}_{h=1}^n$ runs over $NCPP_*(2n)$.\\
As a consequence, we obtain:
\begin{align*}
\big\vert \{-1,1\}^{2n}_{+,*}\big\vert=\big\vert \{-1,1\}^{2(n-1)}_{+}\big\vert=\big\vert NCPP_*(2n) \big\vert=\big\vert NCPP(2(n-1)) \big\vert=C_{n-1}
\end{align*}
Hereinafter, 

$\bullet$ $C_m$ denotes the $m-$th Catalan number for any $m\in{\mathbb N}$;

$\bullet$ $\big\vert S\big\vert$ denotes the cardinality of the set $S$.
\end{remark}

The subsequent results, namely Proposition \ref{q2pr12} and Corollary \ref{q2pr18f}, provide elementary yet fundamental properties of the $({\bf q},2)-$Fock space and the creation--annihilation on it. The proofs will be omitted and can be found in \cite{YGLu2023c} (as well as \cite{YGLu2023a}).

\begin{proposition}\label{q2pr12}Let $\mathcal{H}$ be a Hilbert space  and let $q\in[-1,1]$. For any $f\in\mathcal{H}$, the $({\bf q},2)-$creation operator $A^+(f)$ is bounded:
\begin{align*}
\Vert A^+(f)\Vert =\Vert f\Vert\cdot
\begin{cases} \sqrt{1+q},&\text{ if }q\in[0,1];\\ 1,&\text{ if }q\in[-1,0)  \end{cases}
\end{align*}
Moreover,

1) the {\bf $({\bf q},2)-$annihilation operator} (with the test function $f$) $A(f):=\big(A^+(f) \big)^*$ is well--defined and for any $n\in\mathbb{N}^*$, $\{g_1,\ldots,g_n\}\subset \mathcal{H}$, the following properties hold:
\begin{align*}
A(f)\Phi=0,\quad A(f)(g_1\otimes\ldots\otimes g_n)
=\begin{cases} \langle f, g_1\rangle\Phi, &\text{ if }n=1;\\   \langle f, g_1\rangle g_2+q\langle f, g_2\rangle g_1,&\text{ if }n=2;\\     \langle f, g_1\rangle g_2\otimes\ldots\otimes g_n,&\text{ if }n> 2  \end{cases}
\end{align*}

2) for any $n\in\mathbb{N}^*$ and $f\in\mathcal{H}$,
\begin{align*}
&\Vert A(f)\big\vert_{\mathcal{H}_{1}}\Vert =\Vert A^+(f)\big\vert_{\mathcal{H}_{0}}\Vert =\Vert f\Vert, \quad \Vert A(f)\big\vert_{\mathcal{H}_{n+1}}\Vert =\Vert A^+(f)\big\vert_{\mathcal{H}_n}\Vert\notag\\
&\Vert A(f)A^+(f)\big\vert_{\mathcal{H}_n}\Vert  =\Vert A^+(f)A(f)\big\vert_{\mathcal{H}_n}\Vert =\Vert A(f)\big\vert_{\mathcal{H}_n}\Vert ^2
\end{align*}
and additionally,
\begin{align*}
\Vert A(f)\Vert =\Vert A^+(f)\Vert\,,\quad \Vert A(f)A^+(f)\Vert  =\Vert A^+(f)A(f)\Vert =\Vert A(f)\Vert^2
\end{align*}

3) by denoting, as usual, $A^\epsilon(f):=\begin{cases}
A^+(f),&\text{ if }\epsilon=1\\ A(f),&\text{ if } \epsilon =-1 \end{cases}$ for any $\epsilon\in  \{-1,1\} \text{ and }f\in{\mathcal H}$, the vacuum expectation
\begin{align*}
\big\langle \Phi,A^{\varepsilon(1)}(f_1)\ldots A^{\varepsilon(m)}(f_m)\Phi\big\rangle
\end{align*}
differs from zero only if $m=2n$ (i.e., $m$ is even) and
$\varepsilon\in  \{-1,1\}^{2n}_+$;

4) for any $n\in{\mathbb N}^*$ and $\varepsilon\in  \{-1,1\}^{2n}_+$, and for any $c\in {\mathbb C}$, the following equivalence holds:
\begin{align}
A^{\varepsilon(1)}(f_1)\ldots A^{\varepsilon(2n)} (f_{2n})\Phi=c\Phi\iff\big\langle\Phi,A^{\varepsilon(1)}(f_1)\ldots A^{\varepsilon(2n)}(f_{2n})\Phi\big\rangle=c
\end{align}

5) the $(q,m)-$CR given by \eqref{(q,m)-comm01a} with $m=2$ and \eqref{(q,m)-comm01b} are satisfied.
\end{proposition}

\begin{corollary}\label{q2pr18f}
Let $m\in\mathbb{N}^*$, $\{ f_1,\ldots,f_m\} \subset\mathcal{H}$ and $\varepsilon\in\{-1,1\} ^m$.

1) If $\sum_{h=p}^m \varepsilon(h)\ge0$ for any $p\in\{1,\ldots,n\}$, then the restriction of the operator $A^{\varepsilon(1)}(f_1)\ldots$ $A^{\varepsilon (m)}(f_m)$ to $\oplus_{r=2}^{\infty}\mathcal{H} _r$ is equal to $a_0^{\varepsilon(1)}(f_1)\ldots a_0^{\varepsilon(m)}(f_m)$ (recall that $a_0^+$ and $a_0$ are the usual {\bf free} creation and annihilation operators respectively).

2) If $\sum_{h=1}^m \varepsilon(h)=0$ (it requires necessarily that $m$ is even), for any $r\in{\mathbb N}$, $0\oplus {\mathcal H}_r\oplus 0$ (in particular, ${\mathbb C}\oplus 0$) is invariant under the action of $A^{\varepsilon(1)}(f_1) \ldots A^{\varepsilon (m)}(f_m)$.

3) In the case of $m=2n$,
\begin{align*}
A^{\varepsilon(1)}(f_1)\ldots A^{\varepsilon(2n)} ( f_{2n})\Phi=\langle A^{\varepsilon(1)}(f_1)\ldots A^{\varepsilon(2n)} ( f_{2n})\rangle\Phi
\end{align*}
and the restriction of the operator $A^{\varepsilon(1)}(f_1)\ldots A^{\varepsilon(2n)} ( f_{2n}) $ to the subspace $\oplus_{r=2}^{\infty} \mathcal{H}_r$ is the multiplier of $\prod_{h=1}^n\langle f_{l^\varepsilon_h}, f_{r^\varepsilon_h}\rangle $.
\end{corollary}

\subsection{The vacuum moments of the field operator}
\label{prq2-s2-2}

From now on, we discuss the moments, the moment-generating function and the distribution of the field operator $Q(f)$ (as well as $Q(f)^2$), all with respect to the vacuum state $\big\langle\Phi, \cdot\Phi\big \rangle$.

To determine the distribution of the field operator $Q(f)$, its boundedness, a property that also applies to both $A(f)$ and $A^+(f)$, simplifies the problem to finding its moments.

When $f=0$, it is trivial that $A(0)=A^+ (0)=0$ and so the distribution of $Q(0)$ is clearly the Dirac measure centred at the original. Therefore, we consider $f\neq0$. As previously mentioned, our initial step is to compute all moments of $Q(f)$.
Furthermore, in light of Proposition \ref{q2pr12}, we observe that all odd moments of field operator $Q(f)$ are equal to zero. Consequently, the distribution of $Q(f)$ is fully determined by its even moments:
\begin{equation}\label{q2pr16}
u_{n}:=\big\langle\Phi, Q(f)^{2n}\Phi\big \rangle =\sum_{\varepsilon\in\{-1,1\}_+^{2n}} \big\langle \Phi,A^{\varepsilon(1)}(f)\ldots A^{\varepsilon(2n) }(f)  \Phi\big\rangle ,\ \ \ n\in\mathbb{N}
\end{equation}
Let's denote furthermore,
\begin{equation}\label{q2pr16b}
v_{0}:=1,\quad v_{n}:=\sum_{\varepsilon\in\{-1,1\}^{2n}_{+,*}} \left\langle \Phi,A^{\varepsilon(1)}( f) \ldots A^{\varepsilon( 2n)}(f) \Phi\right\rangle ,\ \ \ n\in\mathbb{N}^*
\end{equation}

\begin{proposition}\label{q2pr18} $u_n$'s verify the system
\begin{equation}\label{q2pr18b}
u_{n+1}=\sum_{k=1}^{n+1}v_{k}u_{n+1-k}
=\sum_{h=0}^{n}v_{h+1}u_{n-h},\quad \forall n\in\mathbb{N}
\end{equation}
with the initial condition $u_0=1$ and $u_1=\Vert f\Vert^2$. While, $v_n$'s can reformulated as follows:
\begin{align}\label{q2pr18a}
v_{n}=\sum_{\varepsilon\in\{-1,1\} _{+}^{2(n-1) }} \left\langle \Phi,A(f) A^{\varepsilon(1) }(f) \ldots A^{\varepsilon(2n-2)}(f)A^{+}(f) \Phi\right\rangle, \quad \forall n\in\mathbb{N}^*
\end{align}
and in particular, $v_1=\Vert f\Vert ^2$.
\end{proposition}

\begin{proof} The statements $u_0=1$ and $u_1=v_1=\Vert f\Vert^2$ directly follow from their respective definitions given in \eqref{q2pr16} and \eqref{q2pr16b}.

By defining $\varepsilon'(j):= \varepsilon(j+1)$ for all $j\in\{1,\ldots, 2n-2\}$, we observe that $\varepsilon'$ traverses over $\{-1,1\}^{2(n-1)}_{+}$ as $\varepsilon$ varies over  $\{-1,1\}^{2n}_{+,*}$. Consequently, \eqref{q2pr18a} is obtained.

The second equality in (\ref{q2pr18b}) is trivial and we will now prove the first.

For any $\left\{(l_{h},r_{h})\right\}_{h=1}^{n}\in NCPP(2n)$, its non--crossing property guarantees that $r_1$ must be an even number. Therefore, we can partition the set $NCPP(2n)$ into $\bigcup\limits_{k=1}^{n}NCPP_k(2n)$, where for any $k\in\left\{ 1,\ldots,n\right\}$, $NCPP_k(2n)$ is a generalization of $NCPP_*(2n)$ defined as:
\[
NCPP_k(2n):=\left\{  \left\{  (l_{h},r_{h}) \right\}_{h=1}^{n}\in NCPP(2n):r_{1}=2k\right\}
\]
The first equality in (\ref{q2pr18b}) is obtained by observing the following easily checked facts:

$\bullet$ the sets $NCPP_k(2n)$'s are pairwise disjoint;

$\bullet$ for any $k=1,\ldots,n$, as $\{(l_h,r_h)\}_{h=1} ^n$ varies over $NCPP_k(2n)$, its {\it first piece} $\{(l_h,r_h)\}_{h=1}^k$ and the {\it second piece} $\{ (l_h,r_h)\}_{h=k+1}^n$ respectively run over the set $NCPP_*(2k)$ and set of all non--crossing pair partitions of the set $\{ 2k+1,\ldots,2n\}$.    \end{proof}

Proposition \ref{q2pr18} establishes a strong dependence of the moments $u_n$'s on the $v_n$'s. The specific characteristics of the $v_n$'s are a key focus of \cite{YGLu2023d}, which articulates the following assertion:
\begin{theorem}\label{q2pr18p}
For any $ n\in\mathbb{N}$ and $f\in\mathcal{H}$,
\begin{align}\label{q2pr18q}
v_{n+1}&=\Vert f\Vert ^{2(n+1)} \sum_{r=0}^n( 1+q)^r \sum_{\substack{i_1,\ldots,i_r\geq0 \\ i_1+\ldots+i_r=n-r}}C_{i_1}\ldots C_{i_r}\notag\\
&=\Vert f\Vert ^{2(n+1)} \sum_{r=0}^n( 1+q)^r
\frac{r}{2n-r}\binom{2n-r}{n}
\end{align}
Hereinafter, 
\begin{equation}\label{q2pr18r}
\sum_{\substack{i_{1},\ldots,i_{r}\geq0\\i_{1}+\ldots
+i_{r}=n-r}}C_{i_{1}}\ldots C_{i_{r}}\left(  1+q\right) ^{r}\Big\vert_{r=0}:=\delta_{n,0}:=\begin{cases}
1, &\text{if }n=0\\ 0, &\text{if }n\ne0
\end{cases}
\end{equation}
\end{theorem}

\subsection{The moment-generating function of $\mathbf{Q(f) ^2}$}  \label{prq2-s2-4}

Due to the boundedness of the field operator, for any $f\in\mathcal{H},$ there exists a $\epsilon_f>0$ such that the series
\begin{equation*}
\sum_{n=0}^\infty x^nu_n=\sum_{n=0}^\infty x^n\langle\Phi,Q(f)^{2n}\Phi\rangle
\end{equation*}
converges when $x\in(-\epsilon_f,\epsilon_f)$. The function $T_q:(-\epsilon_{f},\epsilon_f)  \longmapsto\mathbb{R}$ defined by the above series is commonly referred to as the moment-generating function of $Q(f)^2$. Consequently, the function $x\longmapsto T_q(x^2)$ is the moment-generating function of $Q( f)$. Moreover, the boundedness of the field operator ensures that the distribution of $Q(f)^2$ is fully determined by its moment-generating function $T_q$.

\begin{theorem}\label{q2pr19} For any $f\in\mathcal {H}$, by denoting $\epsilon:=\min\big\{  \frac{1} {4\Vert f\Vert ^2},\epsilon_f\big\} $, we have:
\begin{equation}\label{q2pr19a}
T_q(x) =\frac{1-q+(1+q)\sqrt{1-4\Vert f\Vert ^2x}} {1-q+(1+q)\sqrt{1-4\Vert f\Vert^2x}-2\Vert f\Vert ^2x}, \qquad \forall x\in( -\epsilon,\epsilon)
\end{equation}
\end{theorem}
\begin{proof} 
It follows from the definition of $T_q$ that
\begin{align*}
T_q(x)=&\sum_{n=0}^\infty x^nu_n =1+\sum_{n=0} ^\infty x^{n+1}u_{n+1}\overset{\eqref{q2pr18b}}{=}1+\sum_{n=0} ^\infty x^{n+1}\sum_{h=0}^n v_{h+1}u_{n-h}\notag\\
=&1+\sum_{h=0}^\infty  x^{h+1} v_{h+1}\sum_{n=h}^\infty x^{n-h}u_{n-h}=1+T_q(x)\sum_{h=0}^\infty x^{h+1}v_{h+1}\\
\overset{\eqref{q2pr18q}}{=} &1 +T_q(x)\sum_{n=0}^\infty x^{n+1}\left\Vert f\right\Vert ^{2\left(  n+1\right)  } \sum_ {r=0}^{n}\left(  1+q\right) ^{r}\sum_{\substack{i_{1}, \ldots,i_{r}\geq0 \\i_{1}+\ldots+i_{r}=n-r}}C_{i_{1}}\ldots C_{i_{r}}
\end{align*}
By employing the convention \eqref{q2pr18r} in the above formula, we find that
\begin{align}\label{q2pr19b}
T_q(x)= 1+x\Vert f\Vert^2T_q(x)\bigg(&1+\sum_{n=1}
^\infty \big(x\Vert f\Vert^2\big) ^{n} \sum_ {r=1}^n (1+q)^r\sum_{\substack{i_1,\ldots,i_r\geq0\\ i_1 +\ldots +i_r=n-r}}C_{i_1}\ldots C_{i_r} \bigg)\notag\\
=1+x\Vert f\Vert^2 T_q(x) \bigg( 1+&\sum_{r=1}^{\infty} \big(x\Vert f\Vert^2(1+q)\big) ^r\notag \\
\cdot &\sum_{n=r}^\infty\big(x\Vert f\Vert^2\big)^{ n-r}\sum_{\substack{i_1,\ldots,i_r\geq0\\ i_1+\ldots +i_r =n-r}}C_{i_1}\ldots C_{i_r}\bigg)
\end{align}
Thanks to the following well--known facts:

$\bullet$ the generating function of the Catalan's sequence $\{C_n\}_{n=0}^{\infty}$ is
\begin{equation}\label{q2pr19c}
C(x):=\sum_{n=0}^{\infty}C_nx^n=\frac{1-\sqrt{1-4x}}{2x}=\frac{2}{1+\sqrt{1-4x}},\quad\forall x\in\Big(-\frac{1}{4}, \frac{1}{4}\Big)
\end{equation}

$\bullet$ for any $m\in\mathbb{N}^*$ and $\{  \alpha_{k}\}_{k=0}^{\infty}\subset{\mathbb C}$,
\begin{equation}\label{q2pr19d}
\Big(  \sum_{k=0}^{\infty}\alpha_kx^k\Big)^m= \sum_{k=0}^\infty x^k\sum_{\substack{i_1,\ldots,i_m\ge0 \\ i_1+\ldots+i_m=k}}\alpha_{i_1}\ldots\alpha_{i_m}
\end{equation}\\
we know that for any $x\in\big(  -\frac{1}{4\Vert f\Vert ^2},\frac{1}{4\Vert f\Vert ^{2}}\big) $,
\begin{align*}
&\sum_{n=r}^{\infty}\big(x\Vert f\Vert ^2\big)^{ n-r } \sum_{\substack{i_1,\ldots, i_r\ge0 \\ i_1+\ldots+i_r=n-r}}C_{i_1}\ldots C_{i_r} \overset{k:=n-r}=\sum_{k=0}^{\infty}\big(x\Vert f\Vert ^2\big)^k  \sum_{i_1,\ldots, i_r\ge0 \atop i_1+\ldots+i_r=k}C_{i_1}\ldots C_{i_r} \\
\overset{\eqref{q2pr19d}}{=} &\Big(\sum_{k=0}^\infty \big(x\Vert f\Vert ^2\big)^kC_k\Big)^r
\overset{\eqref{q2pr19c}}{=}\frac{\Big(1-\sqrt{1-4\Vert f\Vert^2x}\big)^r} {\big(2x\Vert f\Vert ^{2}\big)^r}
\end{align*}
Thanks to this equality, (\ref{q2pr19b}) simplifies to the following: for any $x\in (-\epsilon,\epsilon)$, 
\begin{align*}
T_q(x) =&1+x\Vert f\Vert ^2 T_q(x)\bigg( 1+\sum_{r=1} ^{\infty} \big(x\Vert f\Vert^2(1+q)\big)^r
\frac{ \Big(1-\sqrt{1-4\Vert f\Vert^2x}\Big)^r } {\big(2x\Vert f\Vert^2 \big)^r} \bigg)\notag\\
=&1+\frac{2x\Vert f\Vert^2 T_q(x)}{2-(1+q) \big(1-\sqrt{1-4\Vert f\Vert^2x}\big) }
\end{align*}
Thus, \eqref{q2pr19a} is obtained by resolving this equation.    \end{proof}

\begin{remark}\label{q2pr19ax} For $q=0$, one gets
\begin{align*}
T_0(x) &=\frac{1+\sqrt{1-4\Vert f\Vert ^2x}} {1+\sqrt{1-4\Vert f\Vert^2x}-2\Vert f\Vert ^2x}\\
&=
\frac{\big(1+\sqrt{1-4\Vert f\Vert ^2x}\big)\cdot
\big(1-2\Vert f\Vert ^2x-\sqrt{1-4\Vert f\Vert^2x}\big)} {\big(1-2\Vert f\Vert ^2x\big)^2-\big(1-4\Vert f\Vert^2x\big)}=\frac{1-\sqrt{1-4\Vert f\Vert ^2x}}
{2\Vert f\Vert ^2x}
\end{align*}
which exactly presents the moment-generating function of $\xi^2$ with $\xi$ being a random variable distributed in semi--circle law over the interval $( -2\Vert f\Vert ,2\Vert f\Vert )$.
\end{remark}   
%

\section{Distribution of the field operator}
\label{prq2-s3}

We commence by presenting the following simple result, which furnishes the distribution of the field operator for $q\in\{-1,0\}$.

\begin{proposition}\label{q2pr20} For any $0\ne f\in\mathcal{H}$, the distribution of $Q(f)$ is

1) the two-point law on the set $\{-\Vert f\Vert,\Vert f\Vert\}$ with the equi--probability ${1\over2}$ if $q=-1$; 

2) the semi--circle law over the interval $( -2\Vert f\Vert ,2\Vert f\Vert )  $ if  $q=0$.
\end{proposition}
\begin{proof} For $q\in\{-1,0\}$, Theorem \ref{q2pr19} ensures that the moment-generating function of $Q(f)^2$ are, respectively, 
\begin{equation*}
T_{-1}(x)=\frac{1}{1-\Vert f\Vert ^2x};\qquad T_0(x)=
\frac{1-\sqrt{1-4\Vert f\Vert ^2x}}
{2\Vert f\Vert ^2x}
\end{equation*}
Thus, the thesis is established.
%
%
\end{proof}

Moving forward, our attention shifts to deducing the distribution of $Q(f)$ for $q\in(  -1,1]  \setminus\{0\}$. For technical clarity, we opt to introduce a new parameter $a:=1+q$. With this new parameter, the moment-generating function of $Q(f)^2$ assumes the following reformulation of $T_q$:
\[x\longmapsto\frac{2-a+a\sqrt{1-4\Vert f\Vert ^2x}} {2-a+a\sqrt {1-4\Vert f\Vert ^2x}-2\Vert f\Vert ^2x}
\]
Specifically, for $f\in {\mathcal H}$ such that $\Vert f\Vert=\frac{1}{2}$, the aforementioned moment-generating function is
\begin{align}\label{q2pr20d}
S_a(x)  :=&\frac{2-a+a\sqrt{1-4\Vert f\Vert ^2x}} {2-a+a\sqrt {1-4\Vert f\Vert ^2x}-2\Vert f\Vert ^2x}\Bigg\vert_{\Vert f\Vert={1\over2}}\notag\\
=&\frac{16( a-1)  -2(2a^2+a-2)  x+2ax\sqrt{1-x}} {16(a-1)  -4(a-1)(  a+2)  x-x^2}
\end{align}
This function serves as the starting point of our investigation of the distribution of field operator.
In our subsequent discussion, the sphere in $\mathcal{H}$ comprising elements with the norm $\frac{1}{2}$ will play a significant role, denoted as $\mathcal{H}_{1/2}$.

We will commence with some technical preliminaries, followed by an effort to determine the distribution of $Q(  f)^2$ with $f\in\mathcal{H}_{1/2}$. This investigation will be carried out separately in two distinct parts: one for $a\in(1,2]$ (equivalently, $q\in(0, 1]$) and the other for $a\in(  0,1)$ (equivalently, $q\in(-1,0)$). As an application, we will ultimately present the distribution of $Q(  f)$ for any $f$ and $q$. \bigskip

\subsection{Some technical preliminaries}        \label{q2pr-s2-0}

The following measurable function is important for computing the distribution of $Q(f)^2$:
\begin{equation}\label{q2pr20e}
g_a(x)  :=\frac{\sqrt{1-x}}{\sqrt{x}\Big(x^2- \frac{a+2} {4}x-\frac{1}{16(a-1)  }\Big)  }\chi_{(0,1) }( x)
,\quad \forall x\in\mathbb{R}
\end{equation}
where, $a\in (0,1)\cup(1,2]$.

\begin{lemma}\label{q2pr22} Let, for any $a\in (0,1)\cup(1,2]$,
\begin{align}\label{q2pr20h}
h_a(x):=x^{2}-\frac{a+2}{4}x-\frac{1}
{16(a-1)}, \quad \forall x\in \mathbb{R}
\end{align}
then

1) if $a\in (0,1)$, we have $h_a>0$ (so $g_a\ge 0$ and $g_a\Big\vert_{(0,1)}>0$);

2) if $a\in (1,2]$, we have $h_a<0$ on the interval $(0,1)$ (so $g_a\le 0$ and $g_a\Big\vert_{(0,1)}< 0$).
\end{lemma}
\begin{proof}  It is trivial to have
\begin{align*}
 h'_a(x)=2x-{a+2\over4}\in \begin{cases}(0,+\infty),& \text{ if }x>{a+2\over8}\\ (-\infty,0),& \text{ if }x<{a+2\over8}\end{cases}\,;\qquad h'_a(x)\Big\vert_{x={a+2 \over8}}=0
\end{align*}
Thus, $h_a$ decreases on the interval $\big(-\infty, {a+2\over8}\big)$ and increases on the interval $\big({a+2\over8}, +\infty\big)$. Consequently, $h_a$ takes its global minimum at $\frac{a+2}{8}$:
\begin{align}\label{q2pr22a1}
h_a(x)&\ge h_a(x)\Big\vert_{x={a+2\over8}}=
{(a+2)^2\over64}-{(a+2)^2\over32}-{1\over16(a-1)}= {a^3+3a^2\over 64(1-a)},\quad \forall x\in{\mathbb R}
\end{align}
Moreover, since $0<a\le 2$, the global minimum point of $h_a$ (i.e. ${a+2\over8}$) falls within the interval $(0,1)$. Therefore,
\begin{align}\label{q2pr22a2}
h_a(x)< \min\big\{h_a(0),h_a(1)\big\},\qquad \forall x\in(0,1)
\end{align}

In the case of $a\in(0,1)$, \eqref{q2pr22a1} implies:
\begin{align*}
h_a(x)\ge {a^3+3a^2\over 64(1-a)}>0 ,\qquad \forall x\in{\mathbb R}
\end{align*}

For $a\in(1,2]$, it is evident that $h_a(0)=-{1\over16(a-1)}<0$. Combining this fact with \eqref{q2pr22a2}, the thesis (i.e., $h_a(x)<0$ for any $x\in(0,1)$) will be demonstrated if we can show that $h_a(1)\le 0$. This is indeed the case since
\[h_a(1)=1-\frac{a+2}{4}-\frac{1}{16(a-1)} \le 0\iff \frac{1}{16(a-1)} \ge\frac{2-a}{4} \iff 4a^2-12a+9\ge0
\]
and since the function $x\longmapsto 4x^2-12x+9$ reaches its global minimum value zero at $x={3\over 2}$.     \end{proof}\bigskip

We introduce, for any $a>1$,
\begin{equation}\label{q2pr22f}
a_1:=\ \frac{\sqrt{(a+2)^2+\frac{4}{a-1}} +a+2}{8}; \qquad a_2:=\frac{\sqrt{(a+2) ^2+\frac{4}{a-1}}-a-2}{8}
\end{equation}
The condition $a>1$ clearly implies that $\sqrt{(a+2)^2+\frac{4}{a-1}}>a+2$. Consequently,
\begin{equation}\label{q2pr24a}
0<a_2<a_1;\quad a_1a_2=\frac{1}{16(a-1)  };\ \ \ a_1-a_2=\frac{a+2}{4};\quad a_1+a_2=\frac{a\sqrt{a+3}}{4\sqrt{a-1}}
\end{equation}

\begin{remark} As a straightforward corollary of the first two equalities in \eqref{q2pr24a}, we can deduce the following result using the functions $g_a$ and $h_a$ introduced in \eqref{q2pr20e} and \eqref{q2pr20h} respectively:
\begin{align}\label{q2pr24n}
&\frac{1}{a_1+a_2}\left(\frac{1}{a_1-x}+\frac{1}{a_2+x} \right)\sqrt{\frac{1-x}{x}}\chi_{(0,1)}( x) \notag\\
=&\frac{1}{(a_1-x)(a_2+x) }\sqrt{\frac{1-x}{x}} \chi_{(0,1)}( x)  =\frac{1}{a_1a_2+(a_1-a_2)x- x^2 } \sqrt{\frac{1-x}{x}}\chi_{(0,1)}(x)\notag\\
\overset{\eqref{q2pr24a}}{=}&\frac{1}{\frac{1}{16(  a-1) } +\frac{a+2}{4}x- x^2} \sqrt{\frac{1-x}{x}} \chi_{(0,1)}( x)=-\frac{1}{h_a(x)} \sqrt{\frac{1-x}{x}}\chi_{(0,1)}( x)=-g_a(x)
\end{align}
\end{remark}

\begin{lemma}\label{q2pr24}For any $a>1$, the expressions $a_1$ and $a_2$ introduced above possess the following properties:

1) $a_1\ge1$, and the equality holds if and only if $a=\frac{3}{2}$;

2) the following equalities hold:
\begin{equation}\label{q2pr24b}
( a_1a_2+a_1) (a_1a_2-a_2)= \frac{\left(a-\frac{3}{2} \right) ^{2}}{64( a-1)^2}\ ;\ \qquad a_1^2+a_2^2
=\frac{a^2(  a+3)  -2}{16(a-1)  }
\end{equation}
\begin{equation}\label{q2pr24c}
\sqrt{\frac{a_2+1}{a_2}}-\sqrt{\frac{a_1-1}{a_1}}
=\frac{2a\sqrt{( a-1) (a+3)}} {\sqrt{a^2+a-\frac{3}{2} +\vert a-\frac{3}{2}\vert }}
\end{equation}
and
\begin{align}\label{q2pr24d}
&\left( a_2\sqrt{\frac{a_1-1}{a_1}} +a_1 \sqrt{\frac{a_2+1}{a_2}}\right) ^2\notag\\
=&\frac{1}{4(a-1)}\begin{cases}a^2(  a-1)  \left( a+\frac{5}{2}\right)^2-2a+3, &\text{ if } a\in\big(1,{3\over2}\big]\\
a^2(  a-1)  \left( a+\frac{5}{2}\right)^2, &\text{ if } a\in\big({3\over2},2\big]\end{cases}
\end{align}
\end{lemma}

\begin{proof} By employing the definition of $a_1$ and recognizing the condition $a-1>0$, one derives the inequality stated in assertion 1), namely $a_1\ge1$, as follows:
\begin{align*}
&0\le \left( a-\frac{3}{2}\right) ^2\,\iff\, -4a^2+12a-8\le1 \,\iff\,  \frac{1}{a-1}\geq8-4a\\
\iff&\, \sqrt{(a+2)^2+\frac{4}{a-1}}>8-(a+2)
\,\iff\, a_1\ge1 
\end{align*}
Moreover, one of the aforementioned  inequalities becomes an equality if and only if the other does as well. Thus, the affirmation 1) is proved. Now we proceed to prove affirmation 2).

The two equalities in (\ref{q2pr24b}) are obtained easily as follows:
\begin{align*}
&(a_1a_2+a_1)(a_1a_2-a_2)=a_1a_2 \big(a_1a_2+a_1-a_2-1\big) \overset{(\ref{q2pr24a})}{=} \frac{\big(a-{3\over2}\big)^2}{64(  a-1)^2 }\\
&a_1^2+a_2^2=( a_1-a_2)^2+2a_1a_2
\overset{(\ref{q2pr24a})}{=}\Big(\frac{a+2}{4}\Big)^2
+\frac{2}{16(a-1)  }=\frac{a^2(  a+3)  -2}{16(a-1)  }
\end{align*}
By combining the first, third and fourth equalities in (\ref{q2pr24a}), one obtains:
\begin{align}\label{q2pr24h}
&\sqrt{\frac{a_2+1}{a_2}}-\sqrt{\frac{a_1-1}{a_1}}
=\frac{\sqrt{a_1(a_2+1)}-\sqrt{a_2(a_1-1)} }
{\sqrt{a_1a_2}}\notag\\
=&\frac{ a_1+a_2}{\sqrt{a_1a_2}\big(\sqrt{a_1a_2+a_1}
+\sqrt{a_1a_2-a_2}\big)}
=\frac{a\sqrt{a+3}}{ \sqrt{\big(\sqrt{a_1a_2 +a_1}+\sqrt{a_1a_2-a_2}\big)^2}}
\end{align}
Moreover, \eqref{q2pr24a} and \eqref{q2pr24b} guarantee that
\begin{align*}
&\big(\sqrt{a_1a_2 +a_1}+ \sqrt{a_1a_2-a_2}\big)^2\notag\\
=&a_1a_2 +a_1+a_1a_2-a_2+2 \sqrt{\big( a_1a_2+a_1\big)\big(a_1a_2-a_2\big)}\notag\\
=& \frac{1}{8(a-1) }+\frac{a+2}{4}+\frac{\vert  a-\frac{3}{2}\vert}{4( a-1) }={1\over4(a-1)} \Big(a^2+a -\frac{3}{2}+\Big\vert  a-\frac{3}{2}\Big\vert\Big)
\end{align*}
Thus, we derive \eqref{q2pr24c} by applying this formula to \eqref{q2pr24h}.

To prove \eqref{q2pr24d}, we know, thanks to the first equality in (\ref{q2pr24b}), that:
\begin{equation}\label{q2pr24g}
\sqrt{a_1a_2(a_1-1)(a_2+1)}=\sqrt{(a_1a_2+a_1)(a_1a_2- a_2)}=\frac{\left\vert a-\frac{3}{2}\right\vert }{8(  a-1)  }
\end{equation}
Additionally, the expression on the left hand side of \eqref{q2pr24d} is equal to:
\[a_1^2+a_2^2+\frac{a_1^2}{a_2}-\frac{a_2^2}{a_1}+
2\sqrt{ a_1a_2 (a_1-1)(a_2+1)}\]
where,

$\bullet$ \eqref{q2pr24g} and the second equality in \eqref{q2pr24b} guarantee that:
\[a_1^2+a_2^2+2\sqrt{ a_1a_2 (a_1-1)(a_2+1)}=\frac{a^2(  a+3)  -2+4\big\vert a-{3\over2}\big\vert} {16(a-1)}\]

$\bullet$ by expressing $a_1^3-a_2^3$ as $(a_1-a_2) (a_1^2+a_1a_2+a_2^2)$ and by combining together the second equality in \eqref{q2pr24b}, the second and third equalities in \eqref{q2pr24a}, we discover that:
\[\frac{a_1^3-a_2^3}{a_1a_2}=\frac{(a_1-a_2) (a_1^2+ a_1a_2+a_2^2)}{a_1a_2}=\frac{(a+2)(a^2(a+3)-1)}{4}
\]
Therefore,
\begin{align*}
\left( a_2\sqrt{\frac{a_1-1}{a_1}} +a_1 \sqrt{\frac{a_2+1}{a_2}}\right) ^2
=\frac{a^2(a-1)\left(a+\frac{5}{2}\right)^2-a+{3\over2} +\left\vert a-{3\over2}\right\vert }{4(a-1)  }
\end{align*}
which is nothing else than the expression on the right-hand side of \eqref{q2pr24d}. \end{proof} \bigskip

Now, we introduce other two expressions:
\begin{equation}\label{q2pr22g}
A_1:=16( a-1) -\frac{2a}{a_1+a_2}\bigg( \sqrt {\frac{a_2+1}{a_2}}-\sqrt{\frac{a_1-1}{a_1}}\bigg)
\end{equation}
and
\begin{equation}\label{q2pr22h}
A_2:=\frac{2a}{a_1+a_2}\bigg( a_2\sqrt{\frac{a_1-1} {a_1}}+a_1\sqrt{\frac{a_2+1}{a_2}}\bigg)
-2\big( 2a^2+a-2 \big)
\end{equation}

\begin{lemma}\label{q2pr29}For any $a>1$, hold the following:
\begin{equation}\label{q2pr24e}
0\leq A_1=\begin{cases}0, & \text{\ \ if\ \ }a\in\left(  1,\frac{3}{2}\right] \\
16(  a-1)  \left( 1-\frac{a} {\sqrt{a^{2}+2a-3}} \right)  &\text{\ \ if\ \ }a\in\left(  \frac{3}{2},2\right]\end{cases}
\end{equation}
and
\begin{equation}\label{q2pr24f}
A_2=a_2A_1
\end{equation}
\end{lemma}

\begin{proof} Firstly, utilizing the definition of $A_1$ (i.e., \eqref{q2pr22g}), in conjunction with the fourth equality in (\ref{q2pr24a}) and the formula (\ref{q2pr24c}), we find
\begin{align*}
A_1=16(a-1)-\frac{16a(a-1)}{\sqrt{ a^2+a-\frac{3}{2}+ \left\vert a-\frac{3}{2}\right\vert }}
\end{align*}
So, the equality in (\ref{q2pr24e}) is established by observing that
\[
\sqrt{a^2+a-\frac{3}{2}+\big\vert a-\frac{3}{2}\big\vert }=
\begin{cases}a, & \text{ if }a\in\left(  1,\frac{3}{2}\right] \\ \sqrt{a^2+2a-3}, & \text{ if } a\in\left( \frac{3}{2},2\right]
\end{cases}
\]
As a result, we obtain the inequality in (\ref{q2pr24e}) by noting that $\sqrt{a^2+2a-3}\ge a$ for any $a\in\Big( \frac{3}{2},2\Big]$.

Next, we prove the formula (\ref{q2pr24f}) for $a\in \left(1,\frac{3}{2}\right]  $.
Thanks to \eqref{q2pr22h} and (\ref{q2pr24e}), it is evident that the following equivalences are valid for $a\in\left(  1,\frac{3}{2}\right] $:
\begin{align*}
(\ref{q2pr24f})\text{ holds }\iff A_2=0\iff a\left( a_2\sqrt{\frac{a_1-1}{a_1}}+a_1 \sqrt{\frac{a_2+1}{a_2}}\right)=(  2a^2+a-2)(  a_1+a_2)
\end{align*}
That is to say, owing to the fourth equality in (\ref{q2pr24a}), $A_2=0$ if and only if
\begin{equation}\label{q2pr24k}
\left(a_2\sqrt{\frac{a_1-1}{a_1}}+a_1\sqrt{ \frac{a_2+1}{a_2}}\right)^2=\frac{(2a^2+a-2)^2(a+3)} {16(  a-1)  }
\end{equation}
Noticing that for $a\in\left(1,\frac{3}{2} \right]$, (\ref{q2pr24d}) implies that the expression on the left-hand side of (\ref{q2pr24k}) equals
\[
\frac{4a^2( a-1) \left( a+\frac{5}{2}\right)^2-8a+12} {16(  a-1)  }
\]
and which coincides with the right-hand side of (\ref{q2pr24k}).

Lastly, we establish the formula (\ref{q2pr24f}) for $a\in\left(\frac{3}{2},2\right]  $. In this case, \eqref{q2pr24d} implies
\begin{equation*}
\left( a_2\sqrt{\frac{a_1-1}{a_1}}+a_1 \sqrt{\frac{ a_2+1}{a_2}}\right)^2=\frac{a^2\left( a+{5\over2} \right) ^2}{4} =\frac{a^2(  2a+5)^2}{16}
\end{equation*}
i.e., thanks to the positivity of the above terms,
\begin{equation*}
a_2\sqrt{\frac{a_1-1}{a_1}}+a_1\sqrt{\frac{a_2+1}{a_2}}
=\frac{a(  2a+5)  }{4}
\end{equation*}
Therefore,
\begin{equation}\label{q2pr24m}
A_2 =\frac{2a(2a+5)\sqrt{a-1}}{\sqrt{a+3}}-2(2a^2+a-2)
\end{equation}
On the other hand, \eqref{q2pr22f} and \eqref{q2pr24e} yield
\begin{align*}
a_2A_1=&\frac{\sqrt{( a+2)^2+\frac{4}{a-1}}-(a+2)}{8} \cdot 16(a-1)\cdot\left( 1-\frac{a}{\sqrt{a^2+2a-3}} \right) \\
=&\frac{2}{\sqrt{a+3}}\cdot\left(  a(  2a+5)  \sqrt{a-1}-\sqrt{a+3}(  2a^2+a-2)  \right)
\end{align*}
and which coincides precisely with the expression on the right-hand side of (\ref{q2pr24m}).\end{proof}\bigskip

\subsection{The distribution of $\mathbf{Q(f) ^{2}}$ for $\mathbf{a \in(  1,2]}$  and $\mathbf{ \Vert f\Vert=1/2} $}        \label{q2pr-s2-1}

Now, let's introduce
\begin{equation}\label{q2pr24p}
\mu_a(B) :=\frac{a}{8\pi(a-1)} \int_{B}(-g_{a})( x)  dx+\frac{A_1}{16(a-1)}\delta_{a_1}(B)  ,\qquad \forall B\in\mathcal{B}
\end{equation}
where, $\mathcal{B}$ denotes the Borel $\sigma-$algebra on $\mathbb{R}$, and $\delta_c$ denotes the Dirac measure centred at $c\in{\mathbb R}$.

Thanks to the positivity of the function $-g_a$ as shown in Lemma \ref{q2pr22}, along with the positivity of the scalars $A_1$ and $a-1$ (since we are now considering the case of $a>1$), we conclude that $\mu_a$ is indeed a measure. Additionally, it is crucial to highlight that $\mu_a$ is {\bf absolutely continuous} if and only if $A_1=0$.

\begin{theorem}\label{q2pr26} For any $a\in\left(  1,2\right]$ (i.e., $q\in( 0,1] $) and $f\in\mathcal{H}_{1/2}$, we have:
\begin{equation*}
\int\frac{\mu_a(  dx)  }{1-tx}=S_a(  t),\qquad \forall\ t\in(  -1,1)
\end{equation*}
In other words, $\mu_a$ introduced in \eqref{q2pr24p} is the distribution of $Q(f)^2$.
\end{theorem}
\begin{proof} By utilizing the definitions of $A_1$ and $A_2$, and referring to the following well--known formulae (as found in, for example, Section 17.3.12 of \cite{AbrSte1965}):
\begin{align}\label{q2pr26b}
&\int_0^1\frac{\sqrt{1-x^2}}{1-\alpha x^2}dx=
\frac {\pi}{2}\cdot\frac{1-\sqrt{1-\alpha}}{\alpha}, \quad\forall\alpha\in(  -1,1]\notag\\
&\int_0^1\frac{\sqrt{1-x^2}}{\alpha+x^2}dx=
\frac{\pi}{2}\cdot\Big(\sqrt{\frac{1+\alpha}{\alpha}}-1\Big)  ,\quad \forall\alpha>0
\end{align}
we can derive, for sufficiently small $\vert t\vert$, the following result:
\begin{align}\label{q2pr26e}
&\int_0^1\frac{1} {1-tx}\sqrt{\frac{1-x}{x}}dx \overset{y:=\sqrt{x} }= 2\int_0^1\frac{1} {1-ty^2}\sqrt{1-y^2}dy\notag\\
\overset{\eqref{q2pr26b}}= &{\pi\over t} \big(1-\sqrt{1-t}\big) ={\pi\over 1+\sqrt{1-t}}
\end{align}
Additionally,

$\bullet$ since $0<{1\over a_1}\le 1$, one gets:
\begin{equation}\label{q2pr26f}
\int_0^1\frac{1} {a_1-x}\sqrt{\frac{1-x}{x}}dx
\overset{y:=\sqrt{x} }=\frac{2}{a_1}\int_0^1\frac{1} {1-{y^2\over a_1}}\sqrt{1-y^2}dy\overset{\eqref{q2pr26b}}=\pi \bigg(1-\sqrt{1-{1\over a_1}}\bigg)
\end{equation}
and thus,
\begin{align*}
&\int_0^1\frac{1} {1-tx}\frac{1} {a_1-x} \sqrt{\frac{1-x}{x}}dx=\frac{1} {1-a_1t} \int_0^1
\Big(\frac{1} {a_1-x}- \frac{t} {1-tx}\Big) \sqrt{\frac{1-x}{x}}dx\notag\\
\overset{\eqref{q2pr26e},\eqref{q2pr26f}}=&\frac{\pi} {1-a_1t} \bigg(\sqrt{1-t}-\sqrt{1-{1\over a_1}}\bigg)
\end{align*}

$\bullet$ the fact $a_2>0$ ensures that:
\begin{equation*}
\int_0^1\frac{1} {a_2+x} \sqrt{\frac{1-x}{x}}dx
\overset{y:=\sqrt{x} }=2 \int_0^1\frac{1} {a_2+{y^2}}\sqrt{1-y^2}dy\overset{\eqref{q2pr26b}}=\pi \bigg(\sqrt{1+{1\over a_2}}-1\bigg)
\end{equation*}
and consequently,
\begin{align*}
&\int_0^1\frac{1} {1-tx}\frac{1} {a_2+x} \sqrt{\frac{1-x}{x}}dx=\frac{1} {1+a_2t} \int_{0}^{1}
\Big(\frac{1} {a_2+x}+ \frac{t} {1-tx}\Big) \sqrt{\frac{1-x}{x}}dx\notag\\
\overset{\eqref{q2pr26e},\eqref{q2pr26f}}=&\frac{\pi} {1+a_2t} \bigg(\sqrt{1+{1\over a_2}}-\sqrt{1-t}\bigg)
\end{align*}
Summing up, we find that:
\begin{align}\label{q2pr26c}
&\int_{\mathbb{R}}\frac{-g_a(x)  }{1-tx}dx
=\int_0^1\frac{1}{a_1+a_2}\left(  \frac{1} {a_1-x}+\frac{1} {a_2+x}\right)  {1\over 1-tx} \sqrt{\frac{1-x}{x}} dx \notag\\
=&\frac {\pi}{(1+a_2t)(1-a_1t)(a_1+a_2)}\cdot\notag\\
&\cdot \Big( (1-a_1t)\Big( \sqrt{1+a_2^{-1}} -\sqrt{1-t}\Big)+ (1+a_2t) \Big(\sqrt{1-t}-\sqrt{1-a_1^{-1}}\Big)\Big)
\end{align}

The expression
\[{1\over a_1+a_2}\Big((1-a_1t)\Big(\sqrt{1+a_2^{-1}} -\sqrt{1-t}\Big)+ (1+a_2t) \Big(\sqrt{1-t}-\sqrt{1-a_1^{-1}}\Big)\Big)
\]
is indeed a second-order polynomial in the unknowns $t$ and $\sqrt{1-t}$. In this polynomial,

$\bullet$ the coefficient of $t\sqrt{1-t}$ is obviously 1;

$\bullet$ the coefficient of $\sqrt{1-t}$ is trivially zero;

$\bullet$ the coefficient of $t$ is ${-1\over a_1+a_2} \Big(a_1\sqrt{1+a_2^{-1}}+a_2\sqrt{1- a_1^{-1}}\Big)$, which, according to \eqref{q2pr22h} (i.e., the definition of $A_2$), equals ${-1\over 2a}\big(A_{2}+2\left( 2a^2+a-2\right)\big)$;

$\bullet$ the constant term is $\frac{1}{a_1+a_2}\left( \sqrt{1+a_2^{-1} }-\sqrt{1- a_1^{-1}}\right)$, which, according to \eqref{q2pr22g} (i.e., the definition of $A_1$), equals ${1\over 2a}\big(16( a-1) -A_1\big)$.

By applying these facts to \eqref{q2pr26c}, we conclude that:
\begin{align}\label{q2pr26d}
&\int_{\mathbb{R}}\frac{-g_a( x)  }{1-tx}dx\notag\\
=&\frac {\pi}{2a}\cdot\frac{16(a-1)-2(2a^2+a-2 ) t+ 2at\sqrt{1-t}}{( 1+a_2t)(1-a_1t)}-\frac{\pi}{2a} \cdot \frac{A_1+A_2t}{(1+a_2t)(1-a_1t)  }\notag\\
\overset{\eqref{q2pr24f} }{=}&\frac {\pi}{2a} \cdot\frac{16(a-1)-2(2a^2+a-2) t+2at\sqrt{1-t}}
{( 1+a_2t)(1-a_{1}t)}-\frac{\pi}{2a}\cdot\frac {A_1} {(1-a_1t)  }
\end{align}
Furthermore, as
\begin{align*}
&(  1+a_2t)(1-a_1t)=1-(a_1-a_2)t- a_1a_2t^2\\
\overset{\eqref{q2pr24a}}{=}&1-{a+2\over 4}t-{1\over 16(a-1)}t^2={1\over 16(a-1)}\Big(16(a-1)-4(a-1)(a+2)t -t^2\Big)
\end{align*}
we deduce that the first term in the right-hand side of (\ref{q2pr26d}) is equal to, utilizing \eqref{q2pr20d},
\[
\frac{8\pi(a-1)}{a}\cdot S_a( t)
\]
Therefore, \eqref{q2pr26d} simplifies to
\begin{align*}
\int_{\mathbb{R}}\frac{-g_a( x)  }{1-tx}dx &= \frac{8\pi(a-1)  }{a}\cdot S_a( t)  -\frac{\pi}{2a}\cdot \frac{A_1}{1-a_1t}\notag\\
&=\frac{8\pi(a-1)}{a}\cdot S_a(t) -\frac{\pi A_1}{2a} \int_{\mathbb{R} }{d\delta_{a_1}\over 1-tx}
\end{align*}
which can be rewritten as
\begin{align*}
\int_{\mathbb{R}}\frac{1}{1-tx}\mu_a(dx)=S_a(t)
=\int_{\mathbb{R}}\Big(-g_a(x)\frac{a}{8\pi(a-1)}dx
+\frac{A_1}{16(a-1)}d\delta_{a_1}\Big)
\end{align*}
This clearly validates the assertion. \end{proof} \bigskip

Now let's examine two specific cases: $a=2$ and $a=\frac{3}{2}$ (corresponding to $q=1$ and $q=\frac{1}{2}$, respectively).

For $a=2$, according to \eqref{q2pr22f} and \eqref{q2pr24e}, we obtain the following results:
\[
a_1=\frac{\sqrt{5}+1}{2},\qquad a_2=\frac{\sqrt{5}-1} {2}, \qquad A_1=\frac{16(\sqrt{5}-2)}{\sqrt{5}}
\]
Hence,
\begin{align*}
-g_2(  x)  \overset{\eqref{q2pr24n}}{=}
\frac{\sqrt{1-x}}{\sqrt{x}\left(  \frac{\sqrt{5}+1} {2}-x\right)  \left(  \frac{\sqrt{5}-1}{2}+x\right)  }\chi_{(0,1) }(  x)\ ,\qquad \forall x\in \mathbb{R}
\end{align*}
Theorem \ref{q2pr26} tells us that the distribution of $Q(  f)^2$ for any $f\in\mathcal{H}_{1/2}$ is given by the probability measure $\mu_2$, defined as follows:
\begin{equation*}
\mu_2( B)  :=\int_{B}\frac{-g_2( x)} {4\pi} dx+\frac{\sqrt{5}-2}{\sqrt{5}}\delta_{\frac{\sqrt{5}+1} {2}}(  B)\ ,\qquad \forall B\in\mathcal{B}
\end{equation*}

In the case of $a=\frac{3}{2},$ we find that
\[
a_1=1,\quad a_2=\frac{1}{8}\ ;\ \ \ A_1=0
\]
This means that the distribution of $Q( f)  ^2$ for any $f\in\mathcal{H}_{1/2}$ is {\it absolutely continuous} with the probability density function $-\frac{3}{8\pi}g_{\frac{3}{2}}$, where $-g_{\frac{3}{2}}$ is defined as:
\[
-g_{\frac{3}{2}}(  x)  =\frac{\sqrt{1-x}}{\sqrt{x} (1-x) \left(  x+\frac{1}{8}\right)  }\chi_{(0,1)}(x)\ ,\qquad \forall x\in \mathbb{R}
\]

\subsection{The distribution of $\mathbf{Q(  f)^2}$ for $\mathbf{a\in(  0,1)}$  and $\mathbf{\Vert f\Vert=1/2}$} \label{q2pr-s2-3}

Now, let's delve into the case where $a\in(  0,1)$. 

\begin{theorem}\label{q2pr28} For any $a\in( 0,1)$ (i.e., $q\in(-1,0)$) and $f\in\mathcal{H}_{1/2}$, the distribution of $Q(f)^2$ is absolutely continuous with the probability density function $\frac{a}{8\pi(1-a)  }g_a$, where $g_a$ is the function introduced in \eqref{q2pr20e}.
\end{theorem}
\begin{proof} By using the definition of the function $g_a$, i.e., \eqref{q2pr20e}, and the fact 
\begin{align*}
&\frac{1}{(1-tx)\left(  x^2-\frac{a+2} {4}x-\frac{1}{16(a-1)  }\right)  } \\
=&\frac{1}{1-\frac{\left(  a+2\right)  t}{4}- \frac{t^{2}}{16\left( a-1\right)}} \left(\frac{t^{2}} {1-tx}+\frac{tx+1-\frac{\left(  a+2\right)  t}{4}}
{x^{2}-\frac{a+2}{4}x-\frac{1}{16\left(  a-1\right)  }}\right), \quad\forall x\in(0,1)\text{ and } t\in(-1,1)
\end{align*}
one finds that
\begin{align}\label{q2pr28x}
&\int_\mathbb{R}\frac{g_a(x) }{1-tx}dx=\int_0^1
\frac{\sqrt{1-x}}{\sqrt{x}(1-tx)\left(  x^2-\frac{a+2} {4}x-\frac{1}{16(a-1)  }\right)  }dx\notag\\
=& \frac{1}{1-\frac{\left(  a+2\right)  t}{4}- \frac{t^{2}}{16\left( a-1\right)}} \bigg(t^2\int_0^1\frac{\sqrt{1-x}}{\sqrt{x}(1-tx)}dx
+t\int_0^1\frac{\sqrt{x(1-x)}}{x^2-\frac{a+2} {4}x-\frac{1}{16(a-1)  }}dx\notag\\
&\hspace{4cm}+\Big(1-\frac{(a+2)t}{4}\Big)\int_0^1
\frac{\sqrt{1-x}}{\sqrt{x}\Big(x^2-\frac{a+2} {4}x-\frac{1}{16(a-1)  }\Big)}dx\bigg)
\end{align}
With the assistance of formula \eqref{q2pr26b}, an elementary calculation shows that for any $a\in(  0,1)$,
\[
\int_0^1\frac{\sqrt{x(  1-x)  }}{x^2- \frac{a+2}{4}x-\frac{1}{16(a-1)  }}dx=2\pi\left(  \frac{1}{a}-1\right)
\]
and
\[
\int_0^1\frac{\sqrt{1-x}}{\sqrt{x}\left(  x^2-\frac{a+2} {4}x-\frac{1}{16(a-1)  }\right)  }dx= 8\pi\left(  \frac{1}{a}-1\right)
\]
By utilizing these equalities and (\ref{q2pr26b}) to \eqref{q2pr28x}, we can conclude that, for any $a\in(0,1)$ and $t\in(-1,1)$,
\begin{align*}
\int_{0}^{1}\frac{g_{a}\left(  x\right)  }{1-tx}dx=& \frac{\pi t\left(  1-\sqrt{1-t}\right)+
2\pi\left(  \frac{1}{a}-1\right)  t+
8\pi\left(  \frac{1}{a}-1\right) \left(  1-\frac{\left(
a+2\right) t}{4}\right) }{1-\frac{\left(a+2\right)  t}{4}-\frac{t^{2}}{16\left(  a-1\right)  }}\\
=&
\frac{8\pi\left(  1-a\right)  }{a}\cdot\frac{16\left(  a-1\right)  -2\left(2a^{2}+a-2\right)  t+2at\sqrt{1-t}}{16\left(  a-1\right)  -4\left(a-1\right)  \left(  a+2\right)  t-t^{2}}\\ \overset{\eqref{q2pr20d}}=&\frac{8\pi\left( 1-a\right)}
{a}\cdot S_{a}\left(  t\right)
\end{align*}
This clearly confirms the thesis.    \end{proof}

\begin{corollary} For any $a\in(0,2]$, the function $p_a$ defined as
\begin{equation}\label{q2pr28e}
p_a(x):=\frac{2a}{\pi}\cdot\frac{\sqrt{1-x}} {\sqrt{x}\left(1+4(a-1)(a+2)  x-16(  a-1)x^2\right)  }\chi_{(0,1)}(x)  ,\quad \forall x\in\mathbb{R}
\end{equation}
is positive and measurable. Moreover,

1) for $a\in \big(0,{3\over2}\big]$, $p_a$ is the probability density function of the distribution of $Q(f)^2$ with $f\in\mathcal{H}_{1\over 2}$;

2) for any $a\in\left(  \frac{3}{2},2\right]$, $p_a$ does {\bf not} qualify as a probability density function and the distribution of $Q(f)^2$ with $f\in\mathcal{H}_{1\over 2}$ is given by the following probability measure:
\[
B\longmapsto \int_{B}p_{a}(  x) dx+ \frac{A_1}{16(  a-1)  }\delta_{a_1}(  B)  ,\ \ \ \forall B\in\mathcal{B}
\]
\end{corollary}

\begin{proof} Since $p_1(x)=\frac{2}{\pi}\cdot \frac{\sqrt{1-x}} {\sqrt{x}}$ is nothing but the probability density function of $\xi^2$, where $\xi$ distributes in the semi-circle law on the interval $(-1,1)$, Proposition \ref{q2pr20} confirms our result for $a=1$. Now, let's consider the case of $1\ne a\in(0,2]$. In this case, for any $x\in \mathbb{R}$, we have:
\begin{align}\label{q2pr28d}
p_a( x) =\frac{a}{8\pi(1-a)}\frac{\sqrt{1-x}} {\sqrt{x}\left( x^2- \frac{a+2}{4}x-\frac{1}{16(a-1) } \right)  }\chi_{(0,1)}( x)=\frac{a}{8\pi(1-a)}g_a(x)
\end{align}
Additionally, Lemma \ref{q2pr22} guarantees the positivity of the function $p_{a}$ for any $a\in(0,1)\cup (1,2]$. Finally, by utilizing \eqref{q2pr28d}, Theorem \ref{q2pr26}, and Theorem \ref{q2pr28}, we can derive the results in affirmations 1) and 2).      \end{proof}

\subsection{The distribution of $\mathbf{Q( f) }$}
\label{q2pr-s2-4}

\medskip

Now we are ready to present the distribution of the field operator $Q(f) $ with an arbitrary test function $0\neq f\in\mathcal{H}$. 

\begin{theorem}\label{q2pr30} For any $q\in(-1,1]$ (equivalently, $a:=1+q\in(  0,2]  $), and for any $0\ne f\in\mathcal{H}$, we denote

$\bullet$ $L_{q,f}$ as the distribution of the field operator $Q(f)$;

$\bullet$ function $l_{q,f}$ as
\begin{equation}\label{q2pr30a}
l_{q,f}( x) :=\frac{( 1+q)\Vert f\Vert^2}{2\pi} \cdot\frac{\sqrt{4\Vert f\Vert ^2-x^2}} {\left(\Vert f\Vert ^4+q(q+3)\Vert f\Vert^2 x^2-qx^4\right) } \chi_{(  -2\Vert f\Vert ,2\Vert f\Vert )}(  x)  ,\ \ \forall x\in\mathbb{R}
\end{equation}
Then the following statements hold true:

1) for any $q\in\left( -1,\frac{1}{2}\right]$, $L_{q,f}$ is absolutely continuous, and $l_{q,f}$ is its probability density function;

2) for any $q\in\left(  \frac{1}{2},1\right],$ $l_{q,f}$ is {\bf not} probability density function, and in this case,
\begin{equation}\label{q2pr30b}
L_{q,f}(  B)  =\int_{B}l_{q,f}(  x)  dx+ \frac{1}{2}\Big( 1-\frac{1+q}{\sqrt{q( q+4) }} \Big)  \Big(\delta_{-\frac {\sqrt{a_1}}{2\Vert f\Vert ^2}}+\delta_{\frac {\sqrt{a_1}}{2\Vert f\Vert ^2}}\Big) ( B)
\end{equation}
for all Borel sets $B\in\mathcal{B}$, where $a_1$ is defined as in \eqref{q2pr22f}.
\end{theorem}
\begin{remark} It is evident that

$\bullet$ $\frac{1+q}{\sqrt{q( q+4) }} {=}\frac{a}{\sqrt{a^2+2a-3}}$ because $a=1+q$;

$\bullet$ for any $q\in\left(  \frac{1}{2},1\right]  ,$ we have $q( q+4)>q^2+2q+1$, which implies $0<\frac{1+q}{\sqrt{q( q+4) }}<1$.
\end{remark}

\begin{proof}
[Proof of Theorem \ref{q2pr30}] \eqref{q2pr30a} reveals that when $f\in\mathcal{H}_{1/2}$,
\begin{equation*}
l_{q,f}(x)=\frac{2a }{\pi}\cdot \frac{\sqrt{1-x^{2}}} {1+ 4(a-1)(a+2)x^2-16(a-1)x^4 } \chi_{(-1 ,1)}(x),\ \ \forall x\in{\mathbb R}
\end{equation*}
which is indeed the function $x\longmapsto\vert x\vert p_a(x^2)$, where $p_a$ is introduced in \eqref{q2pr28e}. Therefore, for $f\in\mathcal{H}_{1/2}$, the two conclusions of Theorem \ref{q2pr30} are guaranteed by Theorem \ref{q2pr26}, Theorem \ref{q2pr28}, and the following well--known facts:

{\it Let $\xi$ be a 1--dimensional random variable defined on a probability space $\left(  \Omega,\mathcal{F}, \mathbf{P} \right)$. If $\xi$ is {\bf symmetric} (meaning that $\xi$ and $-\xi$ have the same distributed), then its distribution can be determined by the distribution of $\xi^{2}$. Moreover, if additionally, the distribution of $\xi^2 $ can be expressed as a sum $\nu_c+\nu_d$, where

$\bullet$ $\nu_c$ is absolutely continuous with a (sub-probability) density function $f$,

$\bullet$ $\nu_d$ is a discrete measure of the form $\nu_d=\sum_k p_k\delta_{b_k}$, where $p_k\in(0,1]$ and $b_k\in (0,+\infty)$ for all $k$,

\noindent then, the distribution of $\xi$ can be expressed as $\nu'_c+\nu'_d$, where:

$\bullet$ $\nu'_c$ is absolutely continuous with the (sub-probability) density function given by $x\longmapsto \vert x\vert f(x^2)$;

$\bullet$ $\nu'_d$ is a discrete measure of the form $\nu'_d=\sum_k {p_k\over 2}\big(
\delta_{-\sqrt{b_k}}+\delta_{\sqrt{b_k}}\big)$.   }

In general, for any $f\in\mathcal{H}\setminus\{0\} $ and $a\in(0,2]$, we can express $Q(f)$ as $2\Vert f\Vert Q\Big(\frac{f}{2\Vert f\Vert}\Big)$. Therefore, our main result is obtained, as $\frac{f}{2\Vert f\Vert}$ belongs to $\mathcal{H}_{1\over 2}$, by taking into account the following well-known fact:

{\it Let $\eta$ be a 1--dimensional random variable defined on a probability space $\left( \Omega, \mathcal{F}, \mathbf{P}\right)$, and let $\alpha>0$. If the distribution of $\eta $ has the form $\nu_c+\nu_d$, where:

$\bullet$ $\nu_c$ is absolutely continuous with a (sub-probability) density function denoted as $f$,

$\bullet$ $\nu_d$ is a discrete measure of the form $\nu_d=\sum_k p_k\delta_{b_k}$, where $p_k\in(0,1]$ and $b_k\in (0,+\infty)$ for all $k$,

\noindent then the distribution of $\alpha\eta$ follows the form of $\nu'_c+\nu'_d$ and where:

$\bullet$ $\nu'_c$ is absolutely continuous and characterized by the (sub-probability) density function $x\longmapsto {1\over \alpha} f({x\over \alpha})$;

$\bullet$ $\nu'_d$ is a discrete measure of the form $\nu'_d=\sum_k {p_k}\delta_{\alpha b_k}$.   }
\end{proof}

\bigskip

{\bf Acknowledgements}: This work is partially supported by the research fund of University of Bari ``Aldo Moro'': Processi stocastici e loro applicazioni.


\begin{thebibliography}{9}

\bibitem{AbrSte1965}M. Abramowitz, I.A. Stegun: Handbook of Mathematical Functions, Dover Publications (1965).



\bibitem{Ac-Lu96}L. Accardi, Y.G. Lu: The Wigner semi-circle law in quantum electrodynamics. {\it Commun. Math. Phys.}, {\bf 180} (1996), pp.605–632.

\bibitem{aclu-qqbit}L. Accardi, Y.G. Lu: The qq-bit (III) : Symmetric q-Jordan-Wigner embeddings, Infin.
Dimens. Anal. Quantum Probab. Relat. Top., 22, (2019), 1850023.

\bibitem{Ac-Lu2022a}L. Accardi, Y.G. Lu: The quantum moment problem for a classical random variable and a classification of interacting Fock spaces, Infin. Dimens. Anal. Quantum. Probab. Relat. Top., v.25, n.1, 2250003 (2022). DOI:10.1142/S0219025722500035.

\bibitem{acluvo-kyoto}L. Accardi, Y.G. Lu, I. Volovich: The QED Hilbert Module and Interacting Fock Spaces. IIAS--Kyoto report (1997).

\bibitem{AsaYos2020}N. Asai, H.Yoshida: Deformed Gaussian Operators on Weighted $q-$Fock Spaces. Journal of Stochastic Analysis, v.1, n.4, Article 6 (2020). DOI: 10.31390/josa.1.4.06.

\bibitem{Bieden1989}L.C. Biedenharn: The quantum group SUq (2) and a q-analogue of the boson operators, J. Phys. A: Math. Gen., 22 (18), (1989) L873, DOI: 10.1088/0305-4470/22/18/004.

\bibitem{Blitvic}N. Blitvi\'c: The $(q,t)-$Gaussian Process. Journal of Functional Analysis, 263 (2012), pp.3270--3305.


\bibitem{BoKumSpe97}M. Bo\.zejko, B. K\"ummerer, R. Speicher: q-Gaussian processes: Non--Commutative and
classical aspects, Commun. Math. Phys., v.185, (1997), pp.129--154.

\bibitem{Bo-Spe91}M. Bo\.zejko, R. Speicher: An example of a generalized Brownian motion. Commun. Math. Phys. v.137, (1991), pp.519--531.

\bibitem{Bo-Spe92}M. Bo\.zejko, R. Speicher: An example of a generalized Brownian motion II, Quantum Probability and White Noise Analysis (QP--PQ) VII (1992), pp.67--77.



\bibitem{Catalan1887}E. Catalan: Sur les nombres de Segner, Rendiconti del Circolo Matematico di Palermo, 1, (1887), pp.190--201.

\bibitem{DonMar2003}C. Donati-Martin: Stochastic integration with respect to q-Brownian motion, Probab. Theory Related Fields, v.125, n.1, (2003), pp.77--95.

\bibitem{FreLac2003}D.R. French, P.J. Lacrombe: The Catalan number k-fold self-convolution identity:
the original formulation, J. Combin. Math. Combin. Comput., 46 (2003), pp.191–204.

\bibitem{Fris-Bou70}U. Frisch, R. Bourret: Parastochastics, {\it J. Math. Phys.}, {\bf 11} (1970), pp.364--390.

\bibitem{GessLac2001}I. Gessel, P.J. Lacrombe: A forgotten convolution type identity of Catalan: two
hypergeometric proofs, Util. Math., 59 (2001), pp.97-109.

\bibitem{Gre90}O.W. Greenberg: Example of infinite statistics. Physical Review Letters, {\bf 64} (1990), pp.705--708.

\bibitem{Gre91}O.W. Greenberg: Particles with small violations of Fermi or Bose statistics. {\it Phys. Rev. D}, {\bf 43} (1991), pp.4111--4120. 

\bibitem{Gre2001}O.W. Greenberg: Theories of Violation of Statistics. {\it AIP Conference Proceedings}, v. {\bf 545}, Issue {\bf 1} (2001), doi:10.1063/1.1337721 .

\bibitem{GreHil99}O.W. Greenberg, R.C. Hilborn: Quon Statistics for Composite Systems and a Limit on the Violation of the Pauli Principle for Nucleons and Quarks, {\it Phys. Rev. Lett.}, v.{\bf 83}, n.{\bf 22} (1999), pp.4460–4463.


\bibitem{Lac2000}P.J. Lacrombe: A forgotten convolution type identity of Catalan, Util. Math., 57 (2000), pp.65-72.

\bibitem{lu95}Y.G. Lu: On the interacting free Fock space and the deformed Wigner law, {\it Nagoya Math. J.}, v. {\bf 145} (1997), pp.1–28.

\bibitem{YGlu2008}Y.G. Lu: Gaussian type interacting Fock spaces. Infin. Dimens. Anal. Quantum. Probab. Relat. Top., v.11, n.4, pp.475-494 (2008).

\bibitem{YGLu2023a}Y.G. Lu: A New Interacting Fock Space, The Quon Algebra with Operator Parameter and Its Wick's Theorem, submitted, arXiv:2402.18961v1.

\bibitem{YGLu2023c}Y.G. Lu: Some combinatorial aspects of $({\bf q},2)-$Fock space, accepted for publication in Journal of Combinatorics.

\bibitem{YGLu2023d}Y.G. Lu: Weighted Catalan convolution and $({\bf q},2)-$Fock space, submitted, arXiv:2402.19447v1.



\bibitem{KayLadMos2006}P. Kaye, R. Laflamme, M. Mosca: An Introduction to Quantum Computing, Oxford University Press (2006).

\bibitem{KliSch1997}A. Klimyk, K. Schmüdgen: Quantum Groups and Their Representations. Springer-Verlag, Berlin, Heidelberg (1997). https://doi.org/10.1007/978-3-642-60896-4.

\bibitem{Regev2012}A. Regev: A proof of Catalan’s Convolution formula, Integers, v.12, n.5, pp.929--934  (2012). DOI.org/10.1515/integers-2012-0014, arxiv:1109.0571v1 [math.CO] 2 Sep 2011.



\bibitem{Tsallis1988}C. Tsallis: Possible Generalization of Boltzmann-Gibbs Statistics, Journal of Statistical Physics, v.52, pp.479--487 (1988).

\bibitem{Tedford2011}S.J. Tedford: Combinatorial interpretations of convolutions of the Catalan numbers, Integers 11 (2011).

\bibitem{Woit2024}P. Woit: Quantum Theory, Groups and Representations: An Introduction. Preprint, https://www.math.columbia.edu/~woit/QM/qmbook.pdf
\end{thebibliography}
\end{document}